\documentclass[sigconf]{acmart}

\usepackage{amsmath,amsfonts}
\usepackage{algorithmic}
\usepackage{graphicx}
\usepackage{textcomp}
\usepackage{xcolor}
\usepackage{setspace}
\usepackage[ruled,linesnumbered,noend]{algorithm2e}
\usepackage{multirow, enumitem, graphicx, tabularray, listings, tabularx, booktabs}
\usepackage[T1]{fontenc}
\usepackage{float}

\def\ojoin{\setbox0=\hbox{$\Join$}%
  \rule[.10ex]{.25em}{.4pt}\llap{\rule[1.0ex]{.25em}{.4pt}}}
\def\leftouterjoin{\mathbin{\ojoin\mkern-7.5mu\Join}}

\lstdefinestyle{short_code}{
  basicstyle=\small,
  frame=tblr,
  aboveskip=3mm,
  belowskip=3mm,
  showstringspaces=false,
  columns=flexible,
  breaklines=true,
  breakatwhitespace=true,
  tabsize=3
}
\lstdefinestyle{long_code}{
  basicstyle=\fontfamily{pcr}\small\selectfont,
  frame=tblr,
  aboveskip=3mm,
  belowskip=3mm,
  showstringspaces=false,
  columns=fixed,
  breaklines=true,
  breakatwhitespace=true,
  breakindent=6pt,
  tabsize=3,
  belowcaptionskip=3mm
}
\makeatletter
\let\@ORGmakecaption\@makecaption
\long\def\@makecaption#1#2{\@ORGmakecaption{#1}{#2}\vskip\belowcaptionskip\relax}
\makeatother
\newcolumntype{Y}{>{\centering\arraybackslash}X}

\newfloat{lstfloat}{htbp}{lop}
\floatname{lstfloat}{Listing}
\AtBeginDocument{%
  \providecommand\BibTeX{{%
    \normalfont B\kern-0.5em{\scshape i\kern-0.25em b}\kern-0.8em\TeX}}}

\setcopyright{acmlicensed}
\copyrightyear{2024}
\acmYear{2024}
\acmDOI{XXXXXXX.XXXXXXX}

\acmConference[CIKM '24]{33rd ACM International Conference on Information and Knowledge Management}{October 21--25,
  2024}{Boise, Idaho}

\acmISBN{978-1-4503-XXXX-X/18/06}

\begin{document}

\title{ExtGraph: A Fast Extraction Method of User-intended Graphs from a Relational Database}


\author{Jeongho Park}
\affiliation{%
  \institution{GraphAI}
  \city{Daejeon}
  \country{Republic of Korea}}
\email{jhpark@graphai.io}

\author{Geonho Lee}
\affiliation{%
  \institution{KAIST}
  \city{Daejeon}
  \country{Republic of Korea}}
\email{ghlee5084@kaist.ac.kr}  

\author{Min-Soo Kim}
\authornote{Corresponding author.}
\affiliation{%
  \institution{GraphAI}
  \city{Daejeon}
  \country{Republic of Korea}}
\email{minsoo.k@graphai.io}

\renewcommand{\shortauthors}{Jeongho Park, Geonho Lee, and Min-Soo Kim}

\begin{abstract}
Graph analytics is widely used in many fields to analyze various complex patterns.
However, in most cases, important data in companies is stored in RDBMS's, and so, it is necessary to extract graphs from relational databases to perform graph analysis. 
Most of the existing methods do not extract a user-intended graph since it typically requires complex join query processing.
We propose an efficient graph extraction method, \textit{ExtGraph}, which can extract user-intended graphs efficiently by hybrid query processing of outer join and materialized view.
Through experiments using the TPC-DS, DBLP, and IMDB datasets, we have shown that \textit{ExtGraph} outperforms the state-of-the-art methods up to by 2.78x in terms of graph extraction time.
\end{abstract}




\maketitle

\section{Introduction}
\label{sec:introduction}

Graph analytics is suitable for analyzing various complex patterns within highly interconnected data.
It is already widely used in applications such as social media\cite{ali2023social}, finance\cite{sidorov2018graph}, e-commerce\cite{tuteja2021graph},  bio\cite{walsh2020biokg}, and AI services\cite{futia2020integration, tiddi2022knowledge}.
This growing interest has led to active research in graph analytics methods and algorithms, like recommendation queries and fraud detection.

However, the industry primarily utilizes relational databases for data storage and management.
Traditional relational analytics treats each tuple as an independent entity, focusing on aggregations such as COUNT, SUM, and AVERAGE.
In contrast, graph analytics is specialized in analyzing relationships and interactions between individual entities.
Therefore, applying graph analytics to relational data can obtain the results that are difficult to obtain with relational analytics\cite{jindal2014vertexica, chen2008graph, chen2009graph}.

There are two approaches to apply graph analytics to relational data: (1) translating graph analytics queries into equivalent relational queries\cite{jindal2014vertexica, fan2015case, sun2015sqlgraph, steer2017cytosm} and (2) applying graph queries to a graphs extracted from relational data\cite{de2013converting, jain2013graphbuilder, lee2015table2graph, hassan2018extending, perez2015ringo, xirogiannopoulos2017extracting, anzum2021r2gsync}.
In this paper, we focus on the second approach, which has an advantage that can process complex analytics queries at a relatively low cost once the graph is extracted and stored\cite{jin2022making, cheng2019category}. 
For example, Figure~\ref{fig:real_world_example}(a) is a relational database for a recommendation system, and Figure~\ref{fig:real_world_example}(b) shows the relational queries for the system.
The queries extract many relationships including the one between consumers who bought the same product ($Co\text -pur$) and the one between consumers who saw the same promotion ($Same\text -pro$).
The second approach do not need to perform such time-consuming queries, except the initial extraction of those relationships as a graph.

\begin{figure}[htb!]
    \vspace*{-0.3cm}
    \centerline{\includegraphics[width=0.95\linewidth]{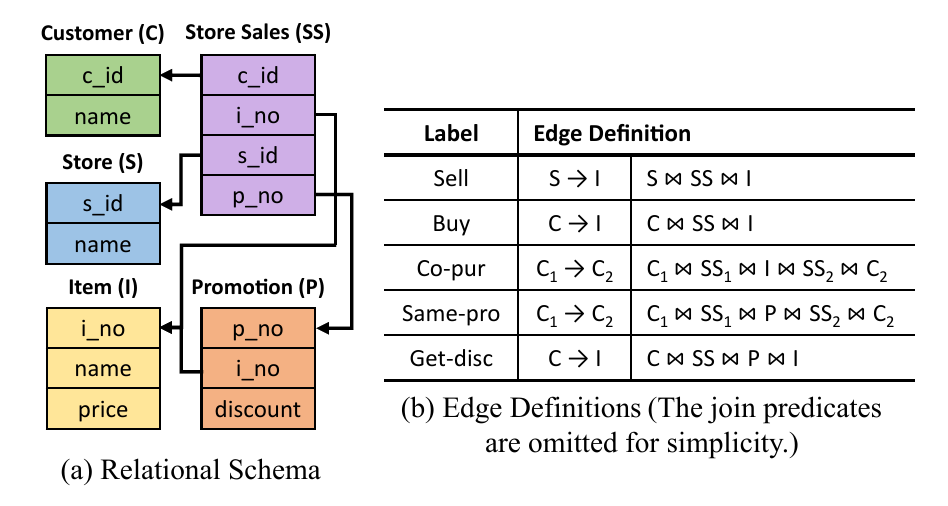}}
    \vspace*{-0.4cm}
    \caption{Example of edge definitions for a relational database. The join predicates are omitted for simplicity.}
    \label{fig:real_world_example}
    \vspace*{-0.3cm}
\end{figure}

Many methods have been proposed for graph extraction, including Table2Graph\cite{lee2015table2graph}, GRFusion\cite{hassan2018grfusion, hassan2018extending}, GraphGen\cite{xirogiannopoulos2015graphgen, xirogiannopoulos2017extracting} and R2GSync\cite{anzum2019graphwrangler, anzum2021r2gsync}, and they can be again categorized into the following two kinds of methods: \textit{schema-based} graph extraction\cite{de2013converting, de2014r2g, jain2013graphbuilder, willke2012graphbuilder, lee2015table2graph, hassan2018extending, hassan2018grfusion} and \textit{join workload-based} graph extraction\cite{perez2015ringo, xirogiannopoulos2015graphgen, xirogiannopoulos2017extracting, anzum2021r2gsync, anzum2019graphwrangler}. 
The \textit{schema-based graph extraction} methods define the edges from the referential relationships between tuples specified by foreign key constraints.
They have a relatively small time cost since the edge definition is straightforward, but they can define only simple graph models\cite{anzum2021r2gsync, xirogiannopoulos2017extracting}.
For example, they cannot extract the five kinds of edges in Figure~\ref{fig:real_world_example}(b).
The \textit{join workload-based graph extraction} methods define the edges from the join queries that express the relationships between tuples.
They are much more general than the former kind of methods in terms of graph model since it can extract an arbitrary graph that a user intends to express with join queries (hereafter denoted as a \textit{user-intended} graph)\cite{anzum2021r2gsync, xirogiannopoulos2017extracting}.
For example, they can define and extract the five kinds of edges in Figure~\ref{fig:real_world_example}(b).
However, they have an issue that can take a very long time to extract a graph if the queries for edge definitions are complex, for example, involving many time-consuming joins.
To alleviate this issue, GraphGen\cite{xirogiannopoulos2015graphgen, xirogiannopoulos2017extracting} and R2GSync\cite{anzum2021r2gsync, anzum2019graphwrangler} decompose the complex queries into simpler and smaller queries, extract the vertices and edges, and then convert them into a graph that consists of \textit{virtual vertices and edges}.
However, their resulting graphs are far from a user-intended graph because a single edge in the user-intended graph is decomposed into multiple vertices and edges in the resulting graph\cite{xirogiannopoulos2017extracting, anzum2021r2gsync}. 
Thus, they need a post-processing step to convert the resulting graph into a user-intended graph, which typically requires a significant amount of time\cite{xirogiannopoulos2017extracting}. 

In this paper, we propose a fast graph extraction method called \textit{ExtGraph} that can define and extract a user-intended graph without the post-processing step.
We observe that there are typically many redundant join operations among the edge definitions for a single user-intended graph and exploit the concept of \textit{join sharing(JS)} to improve the performance in graph extraction. 
In detail, we propose two join sharing techniques: \textit{join sharing by outer join(JS-OJ)} and \textit{join sharing by materialized view(JS-MV)}.
JS-OJ merges multiple queries for the edge definitions into a single outer join query, which implicitly removes the redundant common joins from the queries.
JS-MV explicitly stores the results of some common joins as materialized views and re-uses them.
Then, we propose a hybrid query planning that combines those techniques based on our cost model.

The main contributions of this paper are as follows:
\begin{itemize}
    \item We propose a novel graph extraction method that can extract user-intended graphs having no virtual edges and vertices.
    \item We propose two efficient join sharing techniques, JS-OJ and JS-MV that exploit the concept of join sharing(JS).
    \item We propose the cost models for JS-OJ and JS-MV and a hybrid graph extraction technique based on them.
    \item Through extensive experiments of TPC-DS and real datasets, we have shown ExtGraph significantly improve the graph extraction performance (by up to 2.74X), compared to the state-of-the-art methods.
\end{itemize}

The rest of this paper is organized as follows.
Section \ref{sec:preliminaries} summarizes the existing graph extraction methods and their limitations and the existing optimization techniques about common subexpressions.
In Section~\ref{sec:system_overview}, we explain overall method of ExtGraph.
Then, we propose join sharing in Section~\ref{sec:join_sharing} and present its cost model in Section~\ref{sec:cost_model}.
Section~\ref{sec:experimental_evaluation} presents the experimental results. 
Finally, we discuss related works in Section~\ref{sec:related_works} and conclude this paper in Section~\ref{sec:conclusions}.
\section{Preliminaries}
\label{sec:preliminaries}

\subsection{Problem of Graph Extraction}
\label{sec:graph_extraction}

To extract graphs from relational data, it is necessary to define the graph model to be extracted.
Previous studies define graph models as follows\cite{de2013converting, de2014r2g, jain2013graphbuilder, willke2012graphbuilder, lee2015table2graph, hassan2018extending, hassan2018grfusion, perez2015ringo, anzum2019graphwrangler, neo4jetl}.
\begin{definition}[\textbf{Graph model extracted from relational data}]
\label{def:graph_model_from_rdb}
For a given relational database instance $D = \left\{R_i\right\}$, a graph model is defined as $M=(M_v, M_e)$, where $M_v$ is a set of vertex definitions, and $M_e$ is a set of edge definitions. 
A vertex definition $m_v=(l_v, R_v)\in M_v$ indicates the vertex label $l_v$ and a table $R_v$ where each tuple in $R_v$ defines a vertex of label $l_v$. 
An edge definition $m_e=(l_e, m_{src}, m_{dst}, Q) \in M_e$ indicates an edge label $l_e$, a source vertex definition $m_{src}$ and a destination vertex definition $m_{dst}$ of the edge, and a query $Q$ that represents a join relationship between two vertices. 
Each tuple in the result of query $Q$ defines an edge with label $l_e$ that connects the vertices defined by vertex definitions $m_{src}$ and $m_{dst}$.
\end{definition}

Based on the definition above, the previous studies defines the graph extraction as follows\cite{de2013converting, de2014r2g, jain2013graphbuilder, willke2012graphbuilder, lee2015table2graph, hassan2018extending, hassan2018grfusion, perez2015ringo, anzum2019graphwrangler, neo4jetl}.
\begin{definition}[\textbf{Graph extraction}]
\label{def:graph_extraction}
Given a relational database instance $D = \left\{R_i\right\}$, graph extraction is the process of obtaining a directed multigraph $G=(V, E)$. 
It is composed of the following three steps:
\begin{enumerate}
    \item Define a graph model $M$ for $D$ according to Definition~\ref{def:graph_model_from_rdb}.
    \item Extract a set of vertices and a set of edges from $D$ (based on each $m_v \in M_v$ and $m_e \in M_e$).
    \item Convert the extracted vertices and edges into a graph.
\end{enumerate}
\end{definition}

\subsection{Schema-based Graph Extraction}
\label{sec:schema_based_graph_extraction}

This approach defines a graph model based on the referential relationships between tuples.
It defines edges from two types of referential relationships: \textit{foreign key constraint} and \textit{two foreign keys}\cite{anzum2019graphwrangler, anzum2021r2gsync, neo4jetl}. For the latter type, the direction of the edge is defined by users.
For example, Figure~\ref{fig:schema_based_extraction_example}(b) shows a graph model that can be defined for the relational schema in Figure~\ref{fig:schema_based_extraction_example}(a). 
In Figure~\ref{fig:schema_based_extraction_example}(b), the $s\textunderscore id$, $e\textunderscore id$, and $c\textunderscore id$ edges are defined from foreign key constraints that refer to the $Shipper$, $Employee$, and $Customer$ tables from the $Order$ table, respectively. 
The $OD$ edge is defined from two foreign keys in the $Order\ Detail$ table that refers the relationship between the $Product$ and $Order$ tuples.

\begin{figure}[htb!]
    \vspace*{-0.3cm}
    \centerline{\includegraphics[width=0.95\linewidth]{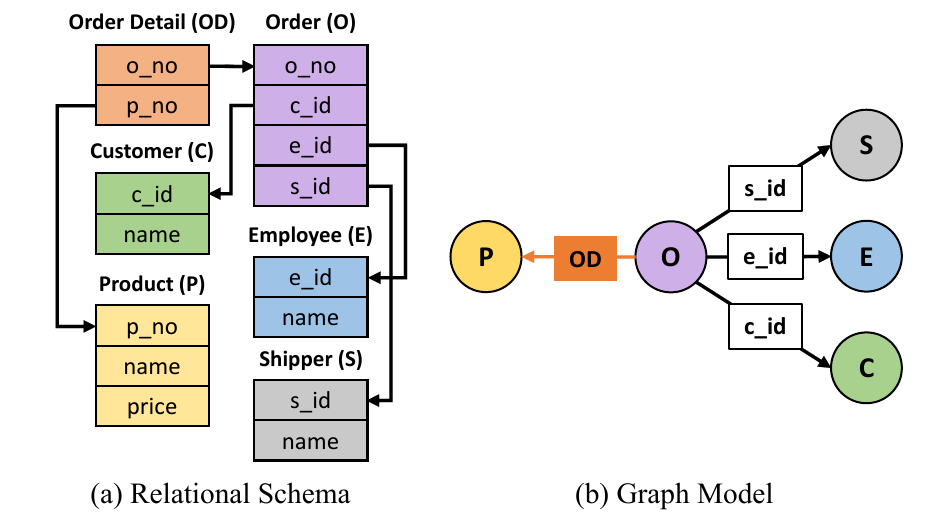}}
    \vspace*{-0.4cm}
    \caption{Example of schema-based graph extraction.}
    \label{fig:schema_based_extraction_example}
    \vspace*{-0.3cm}
\end{figure}

In schema-based graph extraction, the vertices and edges can be retrieved by simple relational operations such as table scan.
For example, it is typically sufficient to scan the $Order$ table for the $s\textunderscore id$, $e\textunderscore id$ and $c\textunderscore id$ edges, and scan the $Order\ Detail$ table for the $OD$ edge, in Figure~\ref{fig:schema_based_extraction_example}.
However, it has a clear limitation on the definable graph model and so cannot define a graph model containing the edges like $Co\text -pur$ or $Same\text -pro$ in Figure~\ref{fig:real_world_example}(b) as mentioned in Section 1.
The method of this approach include  R2G\cite{de2013converting, de2014r2g}, GraphBuilder\cite{jain2013graphbuilder, willke2012graphbuilder}, Table2Graph\cite{lee2015table2graph}, GRFusion\cite{hassan2018extending, hassan2018grfusion} and Neo4j ETL Tool\cite{neo4jetl}.

\subsection{Join Workload-based Graph Extraction}
\label{sec:join_based_graph_extraction}

This approach uses join queries to represent relationships between tuples and defines a graph model based on them. 
Thus, it can define complex relationships such as Figure~\ref{fig:real_world_example}(b) that cannot be modeled with referential relationships. 
It is similar to the concept of \textit{meta path} in heterogeneous graph neural networks\cite{fu2020magnn, sun2011pathsim}.

The method of this approach include Ringo\cite{perez2015ringo}, GraphGen\cite{xirogiannopoulos2015graphgen, xirogiannopoulos2017extracting}, and R2GSync\cite{anzum2021r2gsync, anzum2019graphwrangler}.
Ringo executes each join query for a single edge definition individually and converts the results into a graph. 
This method has a high time cost for graph extraction since it process the queries independently with each other, in particular when they are complex involving many joins or \textit{N}-to-\textit{N} joins between foreign keys. 
For example, $Co\text -pur$ and $Same\text -pro$ in Figure~\ref{fig:real_world_example}(b) are such complex queries involving a join among five tables.

GraphGen\cite{xirogiannopoulos2015graphgen, xirogiannopoulos2017extracting} and R2GSync\cite{anzum2021r2gsync, anzum2019graphwrangler} focus on reducing the time cost of processing the complex join queries. 
They decompose a complex join query into multiple simpler join queries.
Each decomposed query is defined as a \textit{virtual edge}, and a set of virtual edges for a single original join query is stored as a \textit{path}.
Figure~\ref{fig:join_based_extraction_example} shows the graphs extracted by Ringo, GraphGen, and R2GSync for $Co\text -pur$ edge (in blue) and $Same\text -pro$ edge (in orange) in Figure~\ref{fig:real_world_example}(b). GraphGen extracts each $Co\text -pur$ edge as a 2-hop path of virtual edges, each of which represents $C \Join SS \Join I$. R2GSync extracts each $Co\text -pur$ edge as a 4-hop path of virtual edges, each of which represents  $C \Join SS$ or $SS \Join I$. 

\begin{figure}[htb!]
    \vspace*{-0.3cm}
    \centerline{\includegraphics[width=0.95\linewidth]{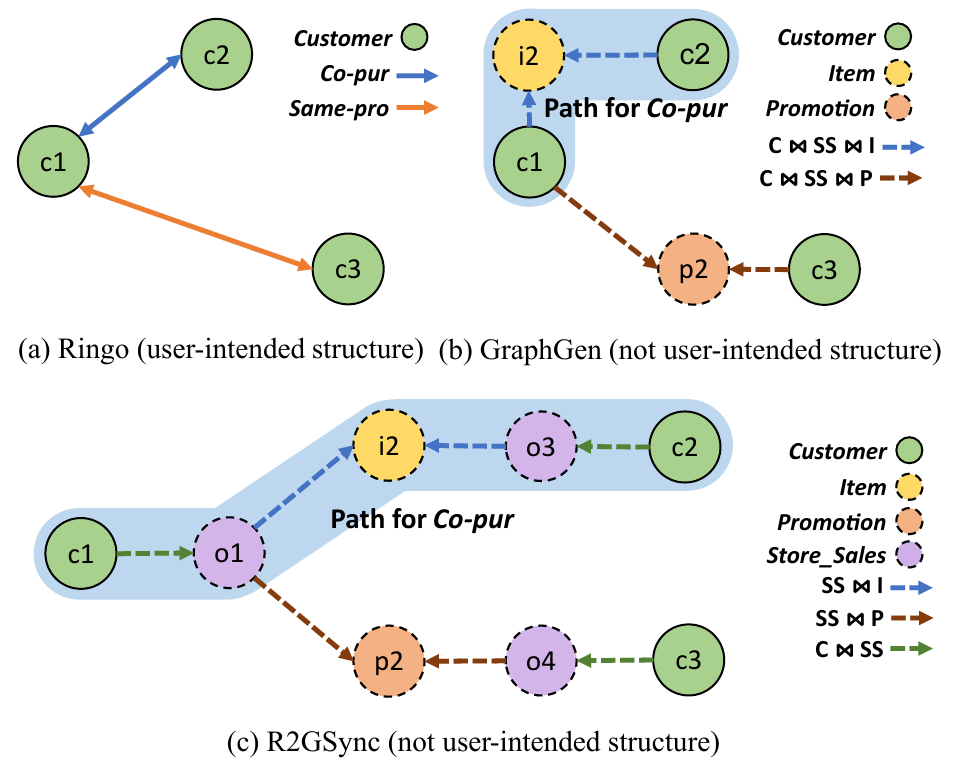}}
    \vspace*{-0.4cm}
    \caption{Examples of join workload-based graph extraction from Figure~\ref{fig:real_world_example}. It only includes the $Co\text -pur$ and $Same\text -pro$ edges for simplicity. The dotted vertices and edges are \textit{virtual vertices and edges} respectively.}
    \label{fig:join_based_extraction_example}
    \vspace*{-0.3cm}
\end{figure}

Each path for $Co\text -pur$ (or $Same\text -pro$) extracted by GraphGen and R2GSync is obviously different from an edge of a user-intended graph extracted by Ringo.
Unless the result graph is converted into a user-intended graph, graph analytics queries (or algorithms) must be modified such that they perform a multi-hop traversal in the result graph (e.g., c1-o1-i2-o3-c2 in Figure~\ref{fig:join_based_extraction_example}(c)) for each single-hop traversal in a user-intended graph (e.g., c1-c2 in Figure~\ref{fig:join_based_extraction_example}(a)).
Nevertheless, converting the resulting graph into a user-intended graph usually requires a substantial amount of time.

Furthermore, some types of edges cannot be extracted by GraphGen and R2GSync.
Both methods support only the edge that can be defined by a \textit{chain} join query like the below one in the Datalog form.
Such edges can be expressed as a path of virtual edges, as shown in Figures~\ref{fig:join_based_extraction_example}(b) and (c).
\begin{lstlisting}[mathescape, style=short_code]
$\displaystyle Q(src, dst):- R_1(src, a_1), R_2(a_1, a_2), R_3(a_2, a_3) \cdots R_n(a_{n-1}, dst)$
\end{lstlisting}
However, they do not support the edges that can be defined by a \textit{star} or \textit{cyclic} queries\cite{anzum2021r2gsync, xirogiannopoulos2017extracting}.
For example, $Get\text- disc$ in Figure~\ref{fig:real_world_example}(b) is an edge defined by a cyclic query as below, where the join among $SS$, $P$, and $I$ tables is cyclic.
\begin{lstlisting}[mathescape, style=short_code]
$\displaystyle Get \text-disc(c\textunderscore id, i\textunderscore no):- $
$C(c\textunderscore id), SS(c\textunderscore id, p\textunderscore no, i\textunderscore no), P(p\textunderscore no, i\textunderscore no), I(i\textunderscore no)$
\end{lstlisting}
GraphGen and R2GSync do not support such edges, while Ringo and ExtGraph do support them.
Table~\ref{tab:graph_model_capability_comparison} summarizes the features of the existing join workload-based methods and our ExtGraph. Ringo has a relatively high cost for graph extraction, but supports extracting a user-intended graph as ExtGraph does. GraphGen and R2GSync have an additional cost for converting a result graph into the corresponding user-intended graph. ExtGraph will be explained in Sections~\ref{sec:system_overview} and \ref{sec:join_sharing}.

\begin{table}[htbp]
\vspace*{-0.1cm}
\caption{Comparison between the existing join workload-based methods and ExtGraph.}
\label{tab:graph_model_capability_comparison}
\vspace*{-0.2cm}
\centering
\begin{tabular}{|c|c|c|c|}
\hline
 & \begin{tabular}[c]{@{}c@{}}Extraction\\ Cost\end{tabular}& \begin{tabular}[c]{@{}c@{}}User-Intended\\ Graph\end{tabular} & \begin{tabular}[c]{@{}c@{}}Converting\\ Cost\end{tabular}\\ \hline
Ringo~\cite{perez2015ringo}     & High   & Yes & N/A \\ \hline
GraphGen~\cite{xirogiannopoulos2015graphgen, xirogiannopoulos2017extracting}  & Low    & No  & High \\ \hline
R2GSync~\cite{anzum2019graphwrangler, anzum2021r2gsync}   & Low    & No  & High \\ \hline
ExtGraph  & Medium & Yes & N/A \\ \hline
\end{tabular}
\vspace*{-0.4cm}
\end{table}

\section{Overview of ExtGraph}
\label{sec:system_overview}

We enhance the process of graph extraction by introducing the optimization using join sharing. Definition~\ref{def:graph_extraction_using_join_sharing} shows the enhanced process, where the new step (Step 2) analyzes the queries defined in Step 1 and optimize the edge definitions using our join sharing techniques (detailed in Section~\ref{sec:join_sharing}).

\begin{definition}[\textbf{Graph extraction using join sharing}]
\label{def:graph_extraction_using_join_sharing}
Given a relational database instance $D$, the process of graph extraction using join sharing follows the four steps: 
\begin{enumerate}
    \item Define a graph model $M$ for $D$ according to Definition~\ref{def:graph_model_from_rdb}.
    \item \textbf{Optimize edge definitions using join sharing techniques.}
    \item Extract a set of vertices and a set of edges from $D$
    \item Convert the extracted vertices and edges into a graph.
\end{enumerate}
\end{definition}

Listing~\ref{fig:graph_model_example} shows an example of the graph model definition for Figure~\ref{fig:real_world_example}(a).
The name of the graph extracted is \textit{RetailG}, which has two types of vertices with the labels $Customer$ and $Item$, and two types of edges with the labels $Get\text-disc$ and $Co\text-pur$.
We omit some vertices and edges in Figure~\ref{fig:real_world_example}(b) due to lack of space. 
In this example, the $Item$ vertex has $name$ and $price$ as properties.
$i\textunderscore no$, PK of the $Item$ table, is also extracted with these properties and used as identifier of the $Item$ vertex.
ExtGraph does not restrict the form of the join queries that define the edges. 
Therefore, it is possible to define edges from star or cyclic queries as well as chain queries. 
In this example, the $Get\text-Disc$ edge is defined from a cyclic query, and the $Co\text-pur$ edge is defined from a chain query. 

\setlength{\textfloatsep}{0.2cm}
\begin{lstfloat}
\vspace*{-0.3cm}
\begin{spacing}{1.1}
\begin{lstlisting}[mathescape, escapeinside={/@}{@/}, style=long_code, caption={Example of a graph model definition.}, label={fig:graph_model_example}]
$\textbf{CREATE GRAPH}$($\textbf{{Graph\textunderscore Name}}$: RetailG);/@\begin{spacing}{0.4}\ \end{spacing}@/
$\textbf{CREATE VERTEX}$($\textbf{{Graph\textunderscore Name}}$: RetailG, 
   $\textbf{{Label}}$: Customer, $\textbf{{ID\textunderscore Column}}$: c_id, 
   $\textbf{{Query}}$: SELECT name from C);/@\begin{spacing}{0.4}\ \end{spacing}@/
$\textbf{CREATE VERTEX}$($\textbf{{Graph\textunderscore Name}}$: RetailG, 
   $\textbf{{Label}}$: Item, $\textbf{{ID\textunderscore Column}}$: i_no, 
   $\textbf{{Query}}$: SELECT name, price from I);/@\begin{spacing}{0.4}\ \end{spacing}@/
$\textbf{CREATE EDGE}$($\textbf{{Graph\textunderscore Name}}$: RetailG, $\textbf{{Label}}$: GetDisc, 
   $\textbf{{Src\textunderscore Label}}$: Customer, $\textbf{{Dst\textunderscore Label}}$: Item, 
   $\textbf{{Query}}$: SELECT null FROM C, SS, P, I 
   WHERE C.c_id=SS.c_id AND I.i_no=SS.i_no 
    AND P.p_no=SS.p_no AND I.i_no=P.i_no);/@\begin{spacing}{0.4}\ \end{spacing}@/
$\textbf{CREATE EDGE}$($\textbf{{Graph\textunderscore Name}}$: RetailG, $\textbf{{Label}}$: CoPur, 
   $\textbf{{Src\textunderscore Label}}$: Customer, $\textbf{{Dst\textunderscore Label}}$: Customer, 
   $\textbf{{Query}}$: SELECT null FROM C1, SS1, I, SS2, C2 
    WHERE C1.c_id=SS1.c_id AND I.i_no=SS1.i_no 
    AND C2.c_id=SS2.c_id AND I.i_no=SS2.i_no);
\end{lstlisting}
\end{spacing}
\vspace*{-0.5cm}
\end{lstfloat}

Figure~\ref{fig:Extracted_Graph_Example} shows the result of the graph extraction with the graph model in Listing~\ref{fig:graph_model_example}. 
As defined in Listing~\ref{fig:graph_model_example}, each tuple in the $Item$ table and the $Customer$ table is extracted as a vertex, and the relationship $Get\text-Disc$ and the $Co\text-pur$ are extracted as edges.
Graph data is more suitable than relational data for exploring relationships between entities, because graph data materializes relationships as edges.
For example, join operations are necessary to find all items that the customer $c1$ purchased with the promotion in the relational data in Figure~\ref{fig:Extracted_Graph_Example}(a).
With the graph in Figure~\ref{fig:Extracted_Graph_Example}(b), finding those items can be done by traversing $Get\text-Disc$ edges, which does not require expensive join operations.

\begin{figure}[htbp]
    \vspace*{-0.3cm}
    \centerline{\includegraphics[width=\linewidth]{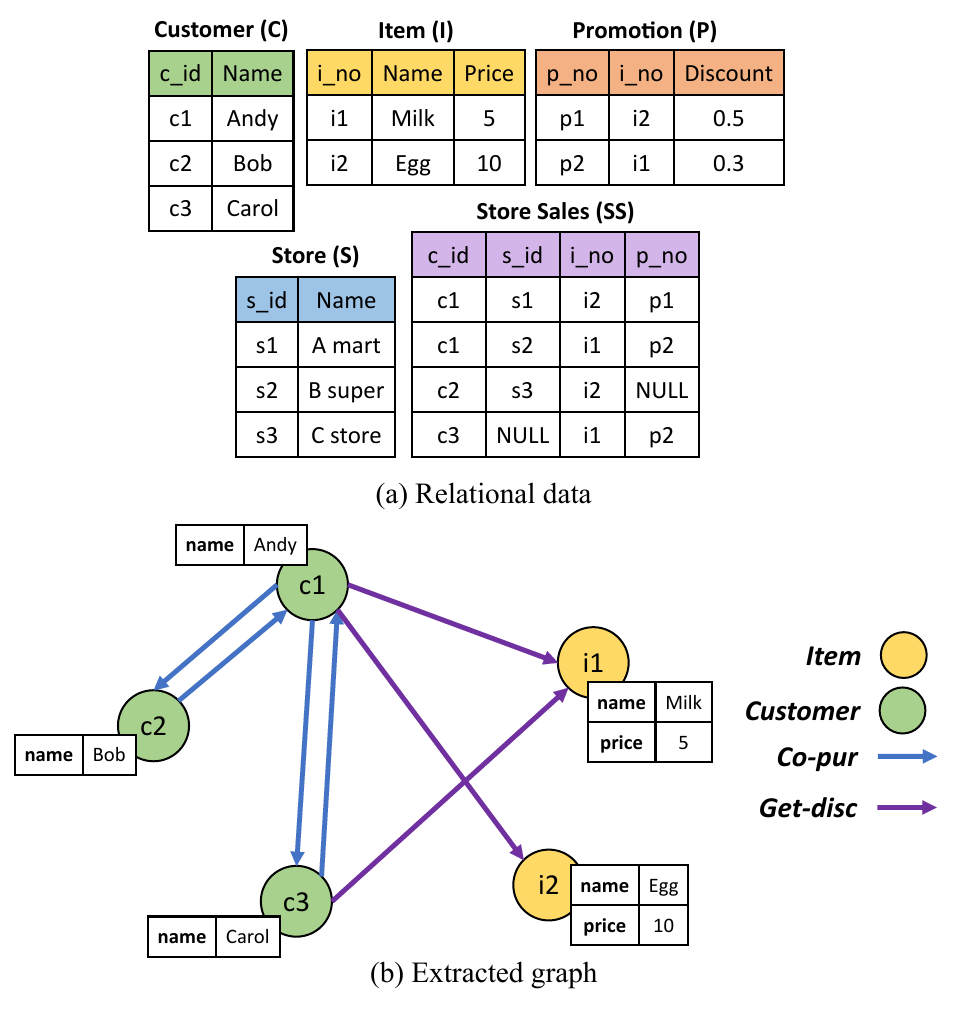}}
    \vspace*{-0.3cm}
    \caption{Example of the extracted graph by graph model in Listing~\ref{fig:graph_model_example}.}
    \label{fig:Extracted_Graph_Example}
    \vspace*{-0.3cm}
\end{figure}
\section{Join Sharing Optimization}
\label{sec:join_sharing}

Our join sharing techniques reduce the redundant sub-queries existing in the queries that define a graph model by using either outer joins or materialized views in order to improve the graph extraction performance.  
We call the technique using outer joins \textit{JS-OJ} and the technique using materialized views \textit{JS-MV}. 
We provide some intuitive explanations for both techniques as below.

Figure~\ref{fig:JS_OJ_motivation}(a) shows the join graphs of the queries that define the $Sell$ and $Buy$ edges in Figure~\ref{fig:real_world_example}(b), and Figure~\ref{fig:JS_OJ_motivation}(b) shows the join graph of the query created by applying JS-OJ to these two queries. 
We note that the original queries only contain the inner joins, while the merged query by JS-OJ contains not only an inner join of $SS \Join I$ but also two outer joins of $SS \leftouterjoin S$ and $SS \leftouterjoin C$.
That is, JS-OJ eliminates a single redundant inner join $SS \Join I$ from the original queries by converting two inner joins ($SS \Join S$, $SS \Join C$) in the original queries to two corresponding outer joins ($SS \leftouterjoin S$, $SS \leftouterjoin C$) in the merged query.
JS-OJ choose the table included in common joins as an outer table.
In this example, $SS$ is the outer table because it is included in common join $SS \Join I$.
Since outer join does not filter out tuples of the outer table, JS-OJ can merge two queries while retrieving the same results as the original queries.
Figure~\ref{fig:JS_OJ_motivation}(c) shows the performance comparison of these queries on the TPC-DS SF=10 database. 
This result shows that the merged query significantly improves the graph extraction performance by using outer joins.

\begin{figure}[htbp]
    \vspace*{-0.3cm}
    \centerline{\includegraphics[width=0.95\linewidth]{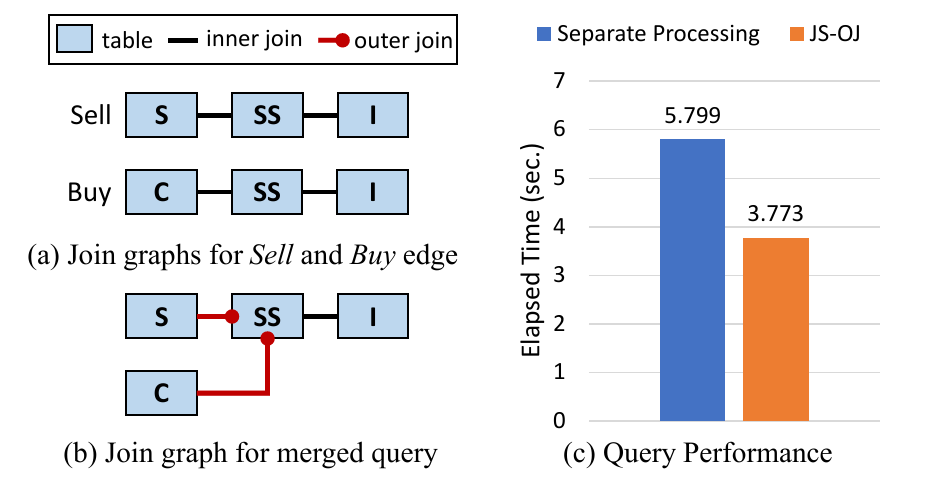}}
    \vspace*{-0.3cm}
    \caption{Example of JS-OJ for $Sell$ and $Buy$ edge in Figure~\ref{fig:real_world_example}(b). For each outer join, the outer table is marked with a circle.}
    \label{fig:JS_OJ_motivation}
    \vspace*{-0.3cm}
\end{figure}

% JS-MV Motivation 설명
% figure로 나타내기
Figure~\ref{fig:JS_MV_motivation}(a) shows the join graphs of the queries that define the $Co\text -pur$ and $Same\text -pro$ edges in Figure~\ref{fig:real_world_example}(b), and Figure~\ref{fig:JS_MV_motivation}(b) shows the join graphs of the queries created by applying JS-MV to these two queries. 
In this case, executing the two queries individually would result in redundant joins because both original queries contain the same join between $C$ and $SS$ four times. 
The JS-MV technique executes the common join $C \Join SS$ is executed first and stores the result in the materialized view $V$. 
Then, it executes the query of $V_1 \Join I \Join V_2$ and the query $V_1 \Join P \Join V_2$ by reusing $V$, which produces the same result with the original queries.
Figure~\ref{fig:JS_MV_motivation}(c) shows the performance comparison of these queries on the TPC-DS SF=10 database. 
The results show that using materialized views significantly improves the overall query processing performance.

\begin{figure}[htbp]
    \vspace*{-0.3cm}
    \centerline{\includegraphics[width=0.95\linewidth]{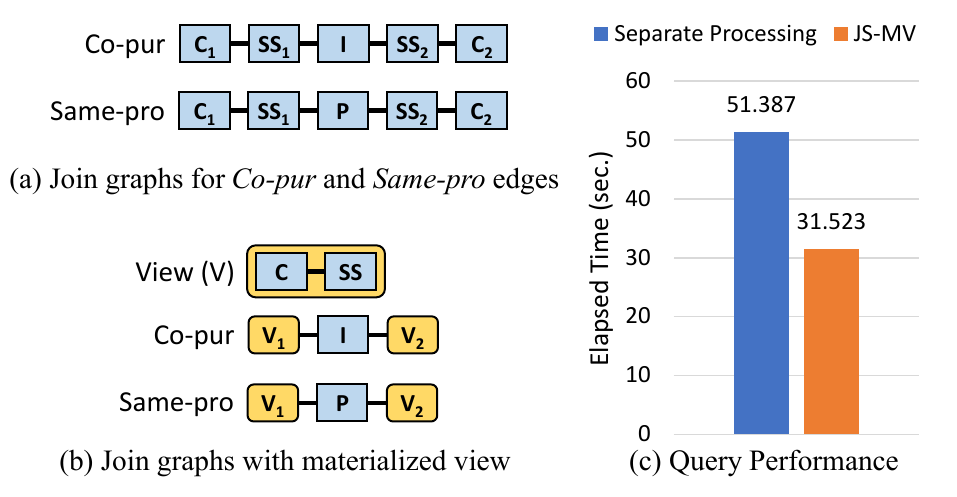}}
    \vspace*{-0.3cm}
    \caption{Example of JS-MV for the join queries for $Co\text -pur$ and $Same\text -pro$ edge from Figure~\ref{fig:real_world_example}(b).}
    \label{fig:JS_MV_motivation}
    \vspace*{-0.3cm}
\end{figure}

When applying join sharing to queries that define a graph, we treat each join query as a join graph. 
The definition of a join graph is as follows.

\begin{definition}[\textbf{Join graph}]
\label{def:join_graph}
A join graph $\mathtt{G = (V, E,f(e \in E),}$ $\mathtt{ g(e \in E))}$ for a join query $Q$ is an undirected multigraph, where $\mathtt{v \in V}$  indicates a table to be joined in Q, and an $\mathtt{e = (}X, Y\mathtt{) \in E}$ indicates a join between two tables $X$ and $Y$. 
There are two labeling functions, $\mathtt{f(e)}$ and $\mathtt{g(e)}$. The $\mathtt{f(e)}$ indicates the type of join, i.e., inner join or outer join, and the $\mathtt{g(e)}$ indicates the join condition for that join, such as $X.i=Y.j$ and $X.i<Y.j$.
\end{definition}

\subsection{Join Sharing by Outer Join (JS-OJ)}
\label{sec:join_sharing_by_outer_join}

To apply JS-OJ to any two queries that contain common joins, we decompose the join graphs of each query $Q_a$ and $Q_b$ based on the common joins. 
We call the subgraph consisting of common joins the \textit{shared subgraph}. 
Its definition is as follows.

\begin{definition}[\textbf{Shared subgraph}]
\label{def:shared_joingraph}
For join queries $Q_a$, $Q_b$ and their join graphs $\mathtt{G_a = (V_a, E_a, f_a(e), g_a(e))}$, $\mathtt{G_b = (V_b, E_b, f_b(e),}$ \linebreak $\mathtt{ g_b(e))}$, A shared subgraph $\mathtt{S = (V_S, E_S), S \subseteq G_a, S \subseteq G_b}$ is a connected component of common joins. 
Thus, For all $\mathtt{e_s \in E_S}$, it satisfies $\mathtt{f_a(e_s) = f_b(e_s)}$ and $\mathtt{g_a(e_s) = g_b(e_s)}$.
\end{definition}

If the join graphs $\mathtt{G_a}$ and $\mathtt{G_b}$ of $Q_a$ and $Q_b$, respectively, contain shared subgraphs, then we can decompose $\mathtt{G_a}$ and $\mathtt{G_b}$ into a single shared subgraph, $S$, and a set of non-shared subgraphs which are disjoint with each other in terms of vertices. 
Here, no edges are connected between the non-shared subgraphs, and each non-shared subgraph has at least one edge connected to a shared subgraph. 
Without loss of generality, there can be multiple possible decompositions for the same queries.
We denote a set of possible decompositions as $\left\{\mathcal{D}_i\right\}$. 
Figure~\ref{fig:JS_OJ_decomposition_example}(a) and (b) show two possible decompositions for given queries $Q_1$ and $Q_2$, where $\mathtt{S}$ indicates a shared subgraph, and $\mathtt{U_1 = \left\{u_{1, i}\right\}}$ and $\mathtt{U_2 = \left\{u_{2, i}\right\}}$ indicate a set of non-shared subgraphs.
Since both queries contain $B \Join D$, we decompose the corresponding subgraph into the shared subgraph. 
Note that we can get different decompositions depending on which subgraph is regarded as a shared subgraph. 
In both cases, the decomposition of $Q_2$ is the same: $\mathtt{S} = \left\{B, D\right\}$, $\mathtt{u_{2,1}}=\left\{E\right\}$, but there are multiple choices for decomposition of $Q_1$. 
For $Q_1$ in $\mathcal{D}_1$, there are a shared subgraph $\mathtt{S} = \left\{B_1, D\right\}$, and three non-shared subgraphs $\mathtt{u_{1,1}}=\left\{A\right\}$, $\mathtt{u_{1,2}}=\left\{C\right\}$, $\mathtt{u_{1, 3}}=\left\{B_2\right\}$.
On the other hand, there are a shared subgraph $\mathtt{S} = \left\{B_2, D\right\}$ and a non-shared subgraph $\mathtt{u_{1,1}}=\left\{A, B_1, C\right\}$ in $\mathcal{D}_2$.

\begin{figure}[htbp]
    \vspace*{-0.4cm}
    \centerline{\includegraphics[width=0.95\linewidth]{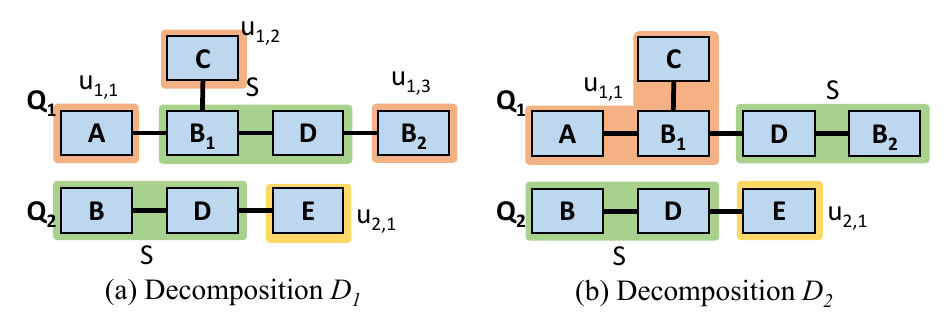}}
    \vspace*{-0.4cm}
    \caption{Examples of query decomposition for JS-OJ.}
    \label{fig:JS_OJ_decomposition_example}
    \vspace*{-0.4cm}
\end{figure}

Algorithm~\ref{alg:JS-OJ} describes how to apply JS-OJ to two queries containing common joins. 
The intuition behind this method is connecting shared subgraphs and non-shared subgraphs with outer joins to form a single join graph. 
First, find all possible decompositions $\left\{\mathcal{D}_i\right\}$ for the two join graphs as shown in Figure~\ref{fig:JS_OJ_decomposition_example} (Line \ref{alg:JS-OJ:decomposition_enumeration}). 
ExtGraph finds the shared subgraphs of the two queries by exhaustive search.
Since the join graph is usually small, the cost of this exhaustive search is trivial.
Then, it finds the combined join graph $\mathtt{G}_{\mathtt{M}i}$ for each partition as shown in Figure~\ref{fig:JS_OJ_result}, and finds the join graph with the cheapest cost among them (Lines \ref{alg:JS-OJ:for_decomposition_begin}-\ref{alg:JS-OJ:for_decomposition_end}). 
The cost model will be discussed in detail in Section~\ref{sec:cost_model}. 
When connecting a shared subgraph and non-shared subgraphs (Lines \ref{alg:JS-OJ:connecting_graph_begin}-\ref{alg:JS-OJ:connecting_graph_end}), convert the edges to outer joins. 
Note that connecting two join queries with an inner join can lead to incorrect join results because of interference between the join operations.
For example, in Figure~\ref{fig:JS_OJ_result}(b), if you connect each subgraph with inner joins, $B_1 \Join D$ will filter the tuples in $D$, which will affect the join result of $E \Join D$.
However, since the two joins are originally independent operations, there should be no interference.
By using outer joins, two queries can be combined into a single query without interference because they do not filter tuples from the outer table, even if there are no matching tuples.

\begin{figure}[htbp]
    \vspace*{-0.3cm}
    \centerline{\includegraphics[width=0.95\linewidth]{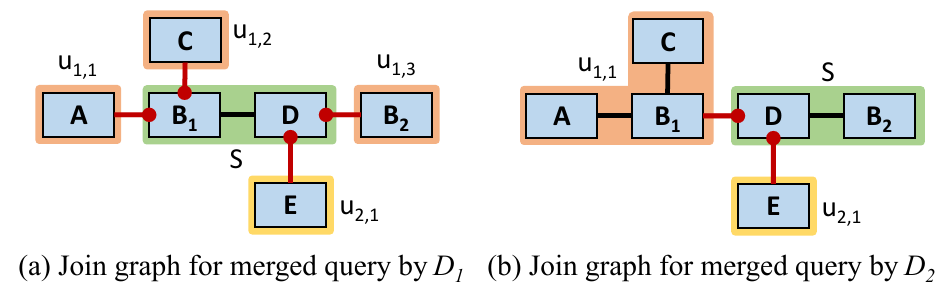}}
    \vspace*{-0.3cm}
    \caption{The resulting queries of JS-OJ for Figure~\ref{fig:JS_OJ_decomposition_example}.}
    \label{fig:JS_OJ_result}
    \vspace*{-0.4cm}
\end{figure}

\setlength{\textfloatsep}{0.2cm}
\begin{algorithm}
\SetAlgoLined
\SetArgSty{textnormal}

\KwIn{two join graphs of given query $\left\{ \mathtt{G_1, G_2} \right\}$}
\KwOut{join graph of merged query $\mathtt{G_M*}$}

$\left\{\mathcal{D}_i\right\} \leftarrow$ find possible decompositions of $\left\{ \mathtt{G_1, G_2} \right\}$; \\
\label{alg:JS-OJ:decomposition_enumeration}
\ForEach{$\mathcal{D}_i \in \left\{ \mathcal{D}_i \right\}$}{
\label{alg:JS-OJ:for_decomposition_begin}
    $\mathtt{S} \leftarrow$ shared subgraph of $\mathtt{G_1, G_2}$ in $\mathcal{D}_i$; \\
    $\mathtt{U_1} \leftarrow$ a set of non-shared subgraphs of $\mathtt{G_1}$ in $\mathcal{D}_i$; \\
    $\mathtt{U_2} \leftarrow$ a set of non-shared subgraphs of $\mathtt{G_2}$ in $\mathcal{D}_i$; \\
    $\mathtt{G}_{\mathtt{M}i} \leftarrow$ copy $\mathtt{S}$ \\
    \ForEach{$\mathtt{u_{1,j}} \in \mathtt{U_1}$}{
    \label{alg:JS-OJ:connecting_graph_begin}
        $\left\{ \mathtt{e_k} \right\} \leftarrow$ find edges of $\mathtt{G_1}$ \\
        \ \ \ \ \ \ \ \ \ \ which connect $\mathtt{S}$ and $\mathtt{u_{1,j}}$;\\
        add all vertices and edges in $\mathtt{u_{1,j}}$ into $\mathtt{G}_{\mathtt{M}i}$;\\
        \ForEach{$\mathtt{e_k} \in \left\{ \mathtt{e_k} \right\}$}{
            mark $\mathtt{e_k}$ as outer join; \\
            add $\mathtt{e_k}$ into $\mathtt{G}_{\mathtt{M}i}$; \\
        }
    }
    \ForEach{$\mathtt{u_{2,j}} \in \mathtt{U_2}$}{
        $\left\{ \mathtt{e_k} \right\} \leftarrow$ find edges of $\mathtt{G_2}$ \\
        \ \ \ \ \ \ \ \ \ \ which connect $\mathtt{S}$ and $\mathtt{u_{2,j}}$;\\
        add all vertices and edges in $\mathtt{u_{2,j}}$ into $\mathtt{G}_{\mathtt{M}i}$;\\
        \ForEach{$\mathtt{e_k} \in \left\{ \mathtt{e_k} \right\}$}{
            mark $\mathtt{e_k}$ as outer join; \\
            add $\mathtt{e_k}$ into $\mathtt{G}_{\mathtt{M}i}$; \\
        }
    }
    \label{alg:JS-OJ:connecting_graph_end}
}
\label{alg:JS-OJ:for_decomposition_end}
$\mathtt{G_M*} \leftarrow$ pick the best plan among $\left\{ \mathtt{G}_{\mathtt{M}i} \right\}$; \\
\Return{$\mathtt{G_M*}$};
\caption{\textsc{JS-OJ}}
\label{alg:JS-OJ}
\end{algorithm}

Theorem~\ref{th:JS_OJ_validation} presents the validity of JS-OJ query result.
The result of query $\mathtt{G_M*}$ includes the results of both the original queries $\mathtt{G_1}$ and $\mathtt{G_2}$. 
All tuples in the results of $\mathtt{G_M*}$ are included in the results of $\mathtt{G_1}$ and $\mathtt{G_2}$, except for those tuples that contain $null$ values due to outer joins. 

\begin{theorem}
\label{th:JS_OJ_validation}
The result of the original queries can be retrieved without loss or error by performing $\mathtt{G_M*}$.
\end{theorem}
\begin{proof}
First, if we represent $\mathtt{G_1}$, $\mathtt{G_2}$, and $\mathtt{G_M*}$ with a shared subgraph $\mathtt{S}$ and non-shared subgraphs $\mathtt{U_1} = \left\{\mathtt{u_{1,j}}\right\}$, $\mathtt{U_2} = \left\{\mathtt{u_{2,j}}\right\}$, then we have the following relational algebra expressions.
\setlength{\textfloatsep}{0cm}
\begin{equation*}
\begin{aligned}
&\mathtt{G_1} = \mathtt{S \Join u_{1, 1} \cdots \Join u_{1, a}}, \mathtt{G_2} = \mathtt{S \Join u_{2, 1} \cdots \Join u_{2, b}} \\
&\mathtt{G_M*} = \mathtt{S \leftouterjoin u_{1, 1} \cdots \leftouterjoin u_{1, a}}\ \mathtt{\leftouterjoin u_{2, 1} \cdots \leftouterjoin u_{2, b}}
\end{aligned}
\end{equation*}

Note that the outer join does not filter the tuples in the outer table, so the following equation holds. 
\setlength{\textfloatsep}{0cm}
\begin{equation*}
\begin{aligned}
\mathtt{S \leftouterjoin u_{1, 1} \cdots \leftouterjoin u_{1, a}}\ \mathtt{\leftouterjoin u_{2, 1} \cdots \leftouterjoin u_{2, b}}\supseteq \mathtt{S \leftouterjoin u_{1, 1} \cdots \leftouterjoin u_{1, a}}
\end{aligned}
\end{equation*}

Furthermore, the following equation holds since the outer join contains all the results of the inner join. 
Therefore, the result of $\mathtt{G_M*}$ contains all the results of $\mathtt{G_1}$.
\setlength{\textfloatsep}{0cm}
\begin{equation*}
\begin{aligned}
\mathtt{S \leftouterjoin u_{1, 1} \cdots \leftouterjoin u_{1, a}} \supseteq 
\mathtt{S \Join u_{1, 1} \cdots \Join u_{1, a}} = \mathtt{G_1}
\end{aligned}
\end{equation*}

Since there is no join operation between non-shared subgraphs, the outer joins of $\mathtt{G_M*}$ can be reordered with the same query result. 
In other words, the following equation holds, which implies that the result of $\mathtt{G_M*}$ contains all the results of $\mathtt{G_2}$.
\setlength{\textfloatsep}{0cm}
\begin{equation*}
\begin{aligned}
\mathtt{G_M*} = & \mathtt{S \leftouterjoin u_{1, 1} \cdots \leftouterjoin u_{1, a}} \ \mathtt{\leftouterjoin u_{2, 1} \cdots \leftouterjoin u_{2, b}} \\
= & \mathtt{S \leftouterjoin u_{2, 1} \cdots \leftouterjoin u_{2, b}} \ \mathtt{\leftouterjoin u_{1, 1} \cdots \leftouterjoin u_{1, a}} \\
\supseteq & \ \mathtt{S \leftouterjoin u_{2, 1} \cdots \leftouterjoin u_{2, b}} \supseteq \mathtt{S \Join u_{2, 1} \cdots \Join u_{2, b}} = \mathtt{G_2}
\end{aligned}
\end{equation*}

Furthermore, combining queries into a single query does not cause interference between non-shared subgraphs because all outer joins occur only between a shared subgraph and non-shared subgraphs, and outer tables are always included in the shared subgraph.
Therefore, a column corresponding to a non-shared subgraph can have a non-$null$ value only if there is a match in the shared subgraph. 
This implies that any tuple with non-$null$ values in the results of $\mathtt{G_M*}$ will also be matched and appear in the results of the original queries, ensuring that none of the results of $\mathtt{G_M*}$ will be invalid.
\end{proof}

\subsection{Join Sharing by Materialized View (JS-MV)}
\label{sec:join_sharing_by_materialized_view}

JS-OJ has the advantage of reducing redundant operations without materialization cost, but it may increase the complexity of the query by combining multiple queries into a single query.
Especially, if the combined single query contains \textit{N}-to-\textit{N} joins that the original queries do not have, it might take more processing time than the original queries.
For example, applying JS-OJ to the queries in Figure~\ref{fig:JS_MV_motivation}(a) results in a query like Figure~\ref{fig:JS_OJ_MV_Comparison}(a). 
We can see that the combined query involves the joins among the fact tables $SS_1$, $SS_2$, and $SS_3$.
Since $SS_1 \Join I \Join SS_2$ and $SS_1 \Join P \Join SS_3$ are already \textit{N}-to-\textit{N} joins, the processing cost of the combined query may significantly increase due to concatenation of those joins with $SS_1$.

\begin{figure}[htbp]
    \vspace*{-0.3cm}
    \centerline{\includegraphics[width=0.95\linewidth]{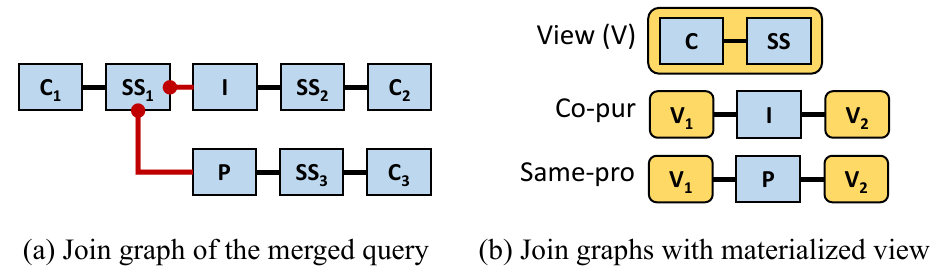}}
    \vspace*{-0.4cm}
    \caption{Comparison of JS-OJ and JS-MV for Figure~\ref{fig:JS_MV_motivation}}
    \label{fig:JS_OJ_MV_Comparison}
    \vspace*{-0.3cm}
\end{figure}

JS-MV incurs the materialization cost, but can reduce redundant operations without increasing the complexity of the queries. 
It stores the result of common joins as materialized views and transforms each query to use those views. 
Since the shared subgraphs represent the common joins across the given queries, the subqueries to be stored as materialized views are obtained via shared subgraphs described in Definition~\ref{def:shared_joingraph}.
The process of transforming given queries to utilize views is simply done by replacing the corresponding parts of the shared subgraph with vertices for the views.
For example, $Co\text -pur$ and $Same\text -pro$ in Figure~\ref{fig:JS_MV_motivation}(a) have a shared subgraph $\mathtt{S} = \left\{C, SS\right\}$ because they have $C \Join SS$ in common. 
Thus, if we store $\mathtt{S}$ as a materialized view $V$ and replace the corresponding parts of each query with $V$, we get the join graph shown in Figure~\ref{fig:JS_OJ_MV_Comparison}(b).

As mentioned in Section~\ref{sec:join_sharing_by_outer_join}, there can be multiple shared subgraphs for the same set of queries.
Depending on which of these shared subgraphs are stored as views and applied to the query, the overall query processing performance could be improved or, in some cases, worsened compared to the original queries.
This is because storing common joins as materialized views is an expensive operation that requires additional I/O.
If the subqueries stored in the view are frequently used in the given queries, the performance improvement from caching will outweigh the materialization cost. 
However, if they are not frequently used, the additional I/O cost will significantly degrade performance.
We use a cost model to determine which subqueries are most efficient to store as views and apply to the given queries. 

\section{Hybrid Optimization of Join Sharing}
\label{sec:cost_model}

We have discussed JS-OJ and JS-MV independently in Section~\ref{sec:join_sharing}, but they can be integrated to further enhance the performance of graph extraction.
For example, in Figure~\ref{fig:cost_model_example}(a), we have three join queries $\left\{Q_1, Q_2, Q_3\right\}$ to extract a graph, and they share two common joins($A \Join B$, $A \Join C$). 
Instead of exclusively applying either JS-OJ or JS-MV to the queries, we can find a more efficient plan for graph extraction by combining both JS-OJ and JS-MV, as in Figure~\ref{fig:cost_model_example}(b), which will be explained in detail in Section~\ref{sec:JS_plan_generation}.
We can integrate both techniques in a cost-based manner.
We describe the cost model of our base system in Section~\ref{sec:base_cost_model}, and the cost models of JS-OJ and JS-MV in Sections~\ref{sec:JS_OJ_cost_model} and \ref{sec:JS_MV_cost_model}, respectively. 
Subsequently, we present the technique for generating a hybrid plan of JS-OJ and JS-MV in Section~\ref{sec:JS_plan_generation}.
Table~\ref{tab:cost_model_symbol} summarizes the symbols used in Section~\ref{sec:cost_model}.

\begin{figure}[tbp]
    % \vspace*{-0.2cm}
    \centerline{\includegraphics[width=0.95\linewidth]{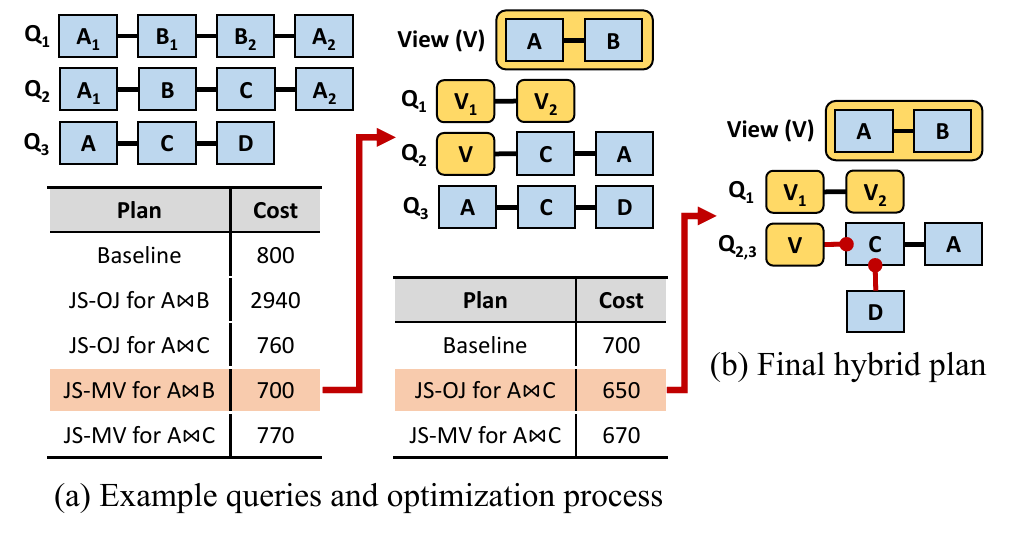}}
    \vspace*{-0.4cm}
    \caption{Example of a hybrid plan of JS-OJ and JS-MV.}
    \label{fig:cost_model_example}
    % \vspace*{-0.2cm}
\end{figure}

\begin{table}[htbp]
\vspace{-0.2cm}
\caption{Summary of symbols.}
\label{tab:cost_model_symbol}
\vspace{-0.2cm}
\begin{tabular}{l|l}
\hline
Symbol        & Description                                  \\ \hline
$Q_i$         & i-th join query for edge definition          \\ \hline
$T_j$         & j-th table in each join query                \\ \hline
$N_P(T)$      & number of disk pages containing data of $T$   \\ \hline
$SQ_{S}$      & subquery for common subgraph $\mathtt{S}$    \\ \hline
$SQ_i$        & subquery for i-th non-common subgraph        \\ \hline
$V_i$         & i-th view definition for JS-MV               \\ \hline
$A_D$         & the cost of accessing a disk page               \\ \hline
\end{tabular}
\vspace{0.5cm}
\end{table}

\subsection{Cost Model of Base System}
\label{sec:base_cost_model}
We consider that there are $N$ join queries for edge definitions of a graph model. 
Eq.~\ref{eq:base_system_cost} shows the processing cost of a query plan $P_{base}$.
\setlength{\textfloatsep}{0cm}
\begin{equation}
\label{eq:base_system_cost}
\begin{aligned}
Cost(P_{base})=\sum_{i=1}^{N}Join(Q_i)
\end{aligned}
\end{equation}
\vspace{-0.3cm}

We assume that only left-deep plan of binary joins and hash join algorithms are used for simplicity.
We also assume that the leftmost table $T_1$ is probed against the remaining tables $\left\{T_2, ... , T_M \right\}$, and the optimal join ordering is found by the base system (e.g., PostgreSQL).
Eq.~\ref{eq:hash_join_cost} shows the cost of each join query.
We omit the detail costs of $Build(T_i)$ and $Probe(T_1)$ for simplicity (see \cite{haas1997s, harris1996join} for details).
\setlength{\textfloatsep}{0cm}
\begin{equation}
\label{eq:hash_join_cost}
\begin{aligned}
Join(Q)=\sum_{i=2}^{M}Build(T_i)+Probe(T_1)
\end{aligned}
\end{equation}
\vspace{-0.4cm}

\subsection{Cost Model of JS-OJ}
\label{sec:JS_OJ_cost_model}

To process the combined query $Q_M$ generated by JS-OJ, we should compute subqueries for the shared subgraph $\mathtt{S}$ and the non-shared subgraphs $\mathtt{U = \left\{u_i\right\}}$, and then compute outer joins on the results of each subquery.
Eq.~\ref{eq:merged_query_cost} shows the processing cost of $Q_M$, where $Outer(O)$ is the cost of processing outer joins, i.e., combining the results of subqueries.
\setlength{\textfloatsep}{0cm}
\begin{equation}
\label{eq:merged_query_cost}
\begin{aligned}
Join(Q_M)=Join(SQ_S)+\sum_{i=1}^{|U|}Join(SQ_i)+Outer(O)
\end{aligned}
\end{equation}
\vspace{-0.3cm}

We assume that the outer join is still evaluated using hash join, similar to the base system, and  the result of the subquery for the shared subgraph is probed against the results of the remaining subqueries.
Thus, Eq.~\ref{eq:outer_join_cost} shows the cost of the outer joins.
\setlength{\textfloatsep}{0cm}
\begin{equation}
\label{eq:outer_join_cost}
\begin{aligned}
Outer(O)=\sum_{i=1}^{|U|}Build(SQ_i)+Probe(SQ_S)
\end{aligned}
\end{equation}
\vspace{-0.4cm}

\subsection{Cost Model of JS-MV}
\label{sec:JS_MV_cost_model}

When JS-MV is applied, the common joins are executed in advance and saved as materialized views.
These materialized views are then utilized during query processing, which plan is denoted as $P_{MV}$.
We assume that a total of $K$ materialized views are created.
Then, the cost of $P_{MV}$ (in Eq.~\ref{eq:JS_MV_cost}) consists of the cost of creating the views and the cost of executing a query $Q'_i$, which is transformed from $Q_i$ to utilize the views.
\setlength{\textfloatsep}{0cm}
\begin{equation}
\label{eq:JS_MV_cost}
\begin{aligned}
Cost(P_{MV})=&\sum_{i=1}^{K}(Join(V_i)+A_D \cdot N_P(V_i))+\sum_{i=1}^{N}Join(Q_i')
\end{aligned}
\end{equation}

If $V_i$ is small, it could potentially be stored only in memory without disk I/O. 
However, for simplicity, we assume that they are always stored on disk in Eq.~\ref{eq:JS_MV_cost}.

\subsection{Hybrid Plan of JS-OJ and JS-MV}
\label{sec:JS_plan_generation}

Algorithm~\ref{alg:JS_optimization} describes a heuristic method for generating a hybrid plan of JS-OJ and JS-MV similar to Figure~\ref{fig:cost_model_example}(b).
For a set of $N$ join queries for graph extraction (called the baseline plan, denoted as $P_base$), it identifies all possible plans by applying either JS-OJ or JS-MV as mentioned in Section~\ref{sec:join_sharing}(Line \ref{alg:JS_optimization:plan_enumeration}), and selects the best plan with the minimum cost (Line \ref{alg:JS_optimization:find_lowest_cost}), denoted as $P_{new}$.
If $P_{new}$ has a lower cost than the current plan $P*$, we replace it with $P_{new}$ (Line \ref{alg:JS_optimization:replace_plan_begin}-\ref{alg:JS_optimization:replace_plan_end}).
We repeat this process until no plan with a lower cost than $P*$ is found.

\vspace{-0.2cm}
\setlength{\textfloatsep}{0.3cm}
\begin{algorithm}
\SetAlgoLined
\SetArgSty{textnormal}
\KwIn{$N$ join queries $\left\{ Q_1, ... , Q_N \right\}$} 
\KwOut{optimal hybrid plan $P*$ for graph extraction} 
$P* \leftarrow \left\{ Q_i \right\}$; \tcp{\small $P*$ is initialized with the baseline plan} 
\label{alg:JS_optimization:baseline_plan}
\While{$true$}{
    $\left\{P_i\right\} \leftarrow $\text{find all plans of applying JS-OJ or JS-MV to $P*$};\\
    \label{alg:JS_optimization:plan_enumeration}
    $P_{new} \leftarrow$ pick the plan with the minimum cost in $\left\{P_i\right\}$ \\
    \label{alg:JS_optimization:find_lowest_cost}
    \If{$Cost(P_{new}) < Cost(P*)$}{
    \label{alg:JS_optimization:replace_plan_begin}
        $P* \leftarrow P_{new}$ \\
    }
    \Else{
        \textbf{break} \\
    }
\label{alg:JS_optimization:replace_plan_end}
}
\Return{$P*$}; 
\caption{\textsc{Hybrid Optimization of Join Sharing}}
\label{alg:JS_optimization}
\end{algorithm}
\vspace{-0.5cm}

For example, in Figure~\ref{fig:cost_model_example}(a), the baseline plan $P_{base}$($P*$ in Line~\ref{alg:JS_optimization:baseline_plan}) has a cost of $800$.
There are four possible plans of JS-OJ or JS-MV($\left\{P_i\right\}$ in Line~\ref{alg:JS_optimization:plan_enumeration}): JS-OJ for $A \Join B$ (cost: 2940), JS-OJ for $A \Join C$  (cost: 760), JS-MV for $A \Join B$ (cost: 700), and JS-MV for $A \Join C$ (cost: 770).
Since applynig JS-MV for $A \Join B$ has the minimum cost among them, it is selected as $P_{new}$(Line~\ref{alg:JS_optimization:find_lowest_cost}), and we replace the baseline with it(Line \ref{alg:JS_optimization:replace_plan_begin}-\ref{alg:JS_optimization:replace_plan_end}).
Similarly, in the next loop, there are two possible plans of JS-OJ or JS-MV($\left\{P_i\right\}$ in Line~\ref{alg:JS_optimization:plan_enumeration}) for the current plan $P*$. 
Applying JS-OJ for $A \Join C$(cost:650) has the minimum cost among them and replaces the baseline with it.
Since there is no possible plans that can obtain by applying JS-OJ or JS-MV, we return the resulting plan $P*$ as in Figure~\ref{fig:cost_model_example}(b).
\section{Experimental Evaluation}
\label{sec:experimental_evaluation}

\subsection{Experimental Setup}
\label{sec:experimental_setup}
\noindent\textbf{Dataset and Graph Definitions: }
For the experiments, we use the TPC-DS\cite{tpcds} dataset (SF 10, 30, 100) and real datasets DBLP\cite{dblp} and IMDB\cite{imdb}. 
Figure~\ref{fig:tpc_ds_graph_models} shows the graph models used in our experiments for TPC-DS, which are commonly used in other studies\cite{anzum2019graphwrangler, anzum2021r2gsync, xirogiannopoulos2017extracting}.
Figure~\ref{fig:tpc_ds_graph_models}(a) is a graph model for a recommendation system, which assumes finds the customers who bought the same product ($Co\text -pur$) and saw the same promotion ($Same\text -pro$) and recommends the products that they bought ($Buy$).
Figure~\ref{fig:tpc_ds_graph_models}(b) is a graph model for fraud detection, which finds a consumer who bought many products from the same store\cite{sahu2020ubiquity} by analyzing a relationship between stores and products ($Sell$) and a relationship between customers and products ($Buy$).
Note that TPC-DS is a benchmark that contains sales data from three channels: \textit{Store}, \textit{Catalog} and \textit{Web}. 
Therefore, we use a total of six graph models for experiments (two scenarios and three channels).

\begin{figure}[htbp]
    \vspace*{-0.2cm}
    \centerline{\includegraphics[width=0.9\linewidth]{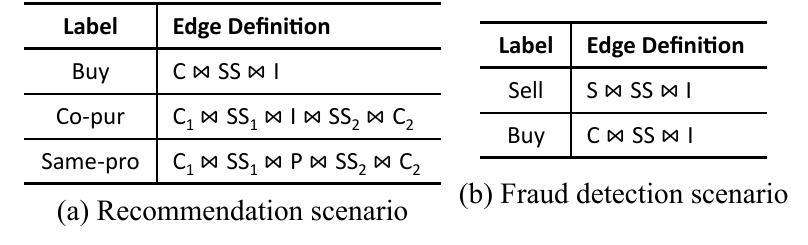}}
    \vspace*{-0.4cm}
    \caption{Graph models for \textit{Store} in TPC-DS (the ones for $Catalog$ and $Web$ are similar).}
    \label{fig:tpc_ds_graph_models}
    \vspace*{-0.3cm}
\end{figure}

\begin{figure}[htbp]
    \vspace*{-0.3cm}
    \centerline{\includegraphics[width=\linewidth]{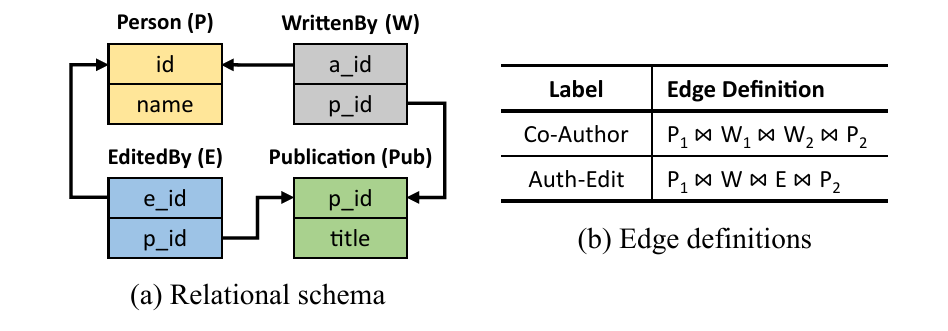}}
    \vspace*{-0.4cm}
    \caption{Graph model used for the DBLP dataset.}
    \label{fig:dblp_graph_models}
    \vspace*{-0.3cm}
\end{figure}

\begin{figure}[htbp]
    \vspace*{-0.3cm}
    \centerline{\includegraphics[width=\linewidth]{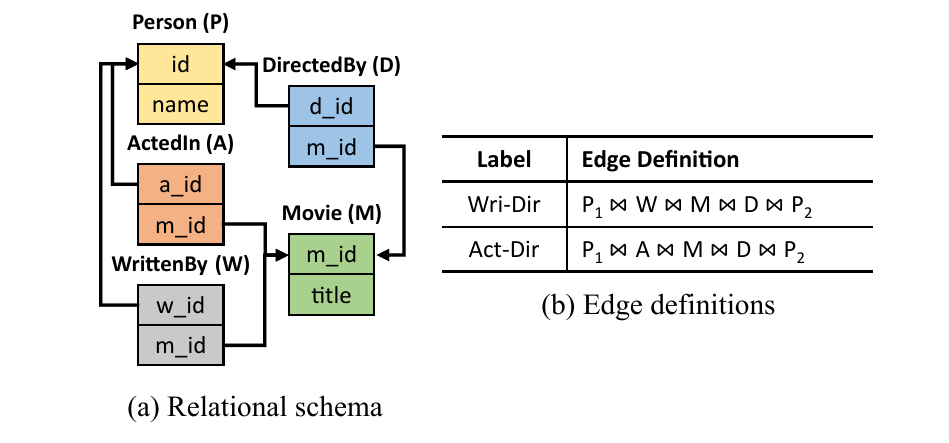}}
    \vspace*{-0.4cm}
    \caption{Graph model used for the IMDB dataset.}
    \label{fig:imdb_graph_models}
    \vspace*{-0.2cm}
\end{figure}

\begin{figure*}[htbp]
    \centerline{\includegraphics[width=\textwidth]{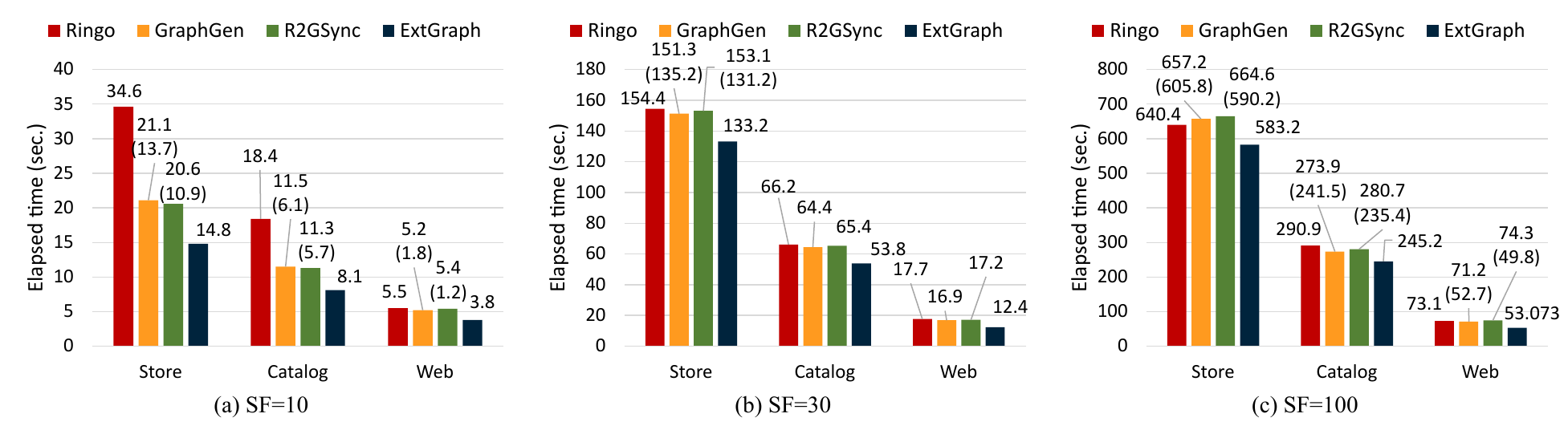}}
    \vspace*{-0.4cm}
    \caption{Performance comparison of extracting \textit{Store}, \textit{Catalog} and \textit{Web} graphs for recommendation scenario (For GraphGen and R2GSync, converting time is specified in parenthesis).}
    \label{fig:recommendation_graph_extraction_performance}
    \vspace*{-0.2cm}
\end{figure*}

\begin{figure*}[htbp]
    \centerline{\includegraphics[width=\textwidth]{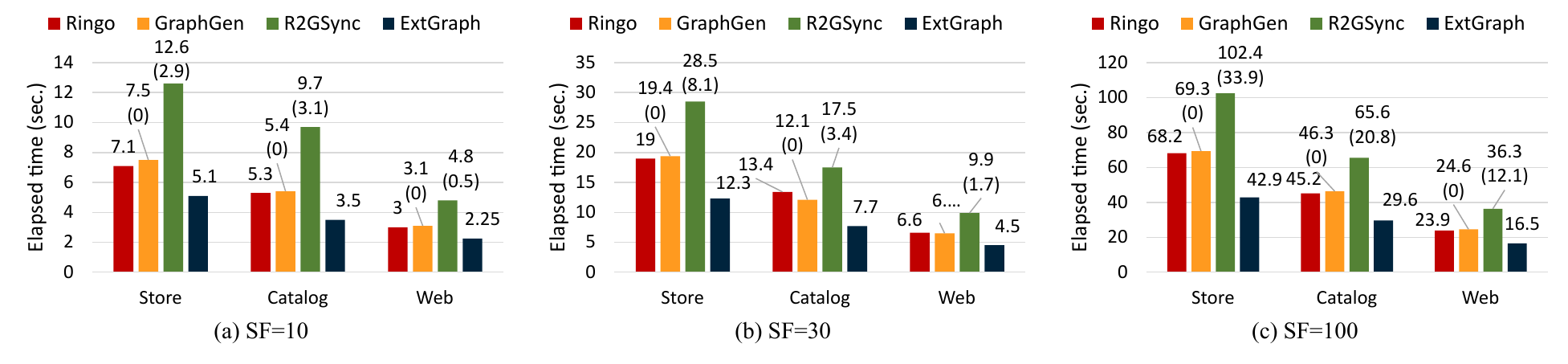}}
    \vspace*{-0.4cm}
    \caption{Performance comparison of extracting \textit{Store}, \textit{Catalog} and \textit{Web} graphs for fraud detection scenario (For GraphGen and R2GSync, converting time is specified in parenthesis).}
    \label{fig:fraud_graph_extraction_performance}
    \vspace*{-0.2cm}
\end{figure*}

Figure~\ref{fig:dblp_graph_models} shows the schema of the DBLP dataset and the graph model, which is also used in other studies\cite{xirogiannopoulos2015graphgen, xirogiannopoulos2017extracting}.
$Auth\text -Edit$ is the relationship between an author and an editor.

Figure~\ref{fig:imdb_graph_models} shows the schema of the IMDB dataset and the graph model, which is also used from other studies\cite{anzum2021r2gsync, xirogiannopoulos2017extracting}.
$Wri\text -Dir$ is the relationship between a writer and a director who worked on the same movie, and $Act\text -Dir$ is the relationship between an actor and a director who worked on the same movie.

\noindent\textbf{Environments: }
We conduct all experiments on a single server with two 16-core 3.0 GHz CPUs, 1 TB of memory, and a 14 TB hard disk. 
The operating system used is Ubuntu 18.04.4.
We implement our method on top of the base system, PostgreSQL 14.4.

\noindent\textbf{Systems Compared: }
Since ExtGraph extracts graphs based on join queries, we used Ringo\cite{perez2015ringo}, GraphGen\cite{xirogiannopoulos2015graphgen, xirogiannopoulos2017extracting}, and R2GSync\cite{anzum2019graphwrangler, anzum2021r2gsync}, the join workload-based graph extraction methods described in Section~\ref{sec:join_based_graph_extraction}, for comparison.
For the fairness of the comparison, we implemented all three methods on top of PostgreSQL extensions like ExtGraph.
In the case of GraphGen and R2Gsync, the extracted graphs are not user-intended graphs as explained in Section~\ref{sec:join_based_graph_extraction}.
Thus, we measure their elapsed times so as to include the time of converting their output to the same user-intended graph like Ringo and ExtGraph.

\vspace{-0.3cm}
\subsection{Graph Extraction Performance}
\label{sec:graph_extraction_performance}

Figure~\ref{fig:recommendation_graph_extraction_performance} and \ref{fig:fraud_graph_extraction_performance} show the elapsed times of graph extraction for two scenarios in Figures~\ref{fig:tpc_ds_graph_models}(a) and \ref{fig:tpc_ds_graph_models}(b), respectively. In both scenarios, ExtGraph demonstrates superior performance compared to the other three methods across all sales channels ($Store$, $Catalog$, $Web$) and scale factors.

In Figure~\ref{fig:recommendation_graph_extraction_performance}, for the $Store$ channel, GraphGen and R2GSync show better performance than Ringo at scale factors SF=10 and SF=30. 
However, Ringo outperforms GraphGen and R2GSync at SF=100. 
GraphGen and R2GSync generate virtual edges as output, which are then converted into user-intended graphs. 
This process resembles JS-MV in terms of materializing the intermediate results (virtual edges) and deriving the final results from them (post-processing). 
The main difference is that our JS-MV materializes only the necessary subqueries (i.e., common joins), whereas GraphGen and R2GSync materialize the entire queries, resulting in additional costs. 
These additional costs are negligible at lower scale factors, allowing GraphGen and R2GSync to exhibit better performance. 
However, at SF=100, these additional costs contribute to performance degradation.
ExtGraph achieves the best performance, in particular outperforms Ringo by 2.34X at SF 10.

In Figure~\ref{fig:fraud_graph_extraction_performance}, we observe that R2GSync exhibits the worst performance. 
This is because R2GSync decomposes the join queries into smaller queries, which has a detrimental impact in this scenario. 
The edges in the fraud detection scenario originate from relatively straightforward relational queries, and breaking them down into smaller queries is less efficient and results in significant post-processing costs.
In contrast, GraphGen achieves performance similar to Ringo because GraphGen does not decompose the queries for the edges at a per-join level in this scenario.
Instead, it decomposes them based on costly joins, avoiding the need for post-processing.
As a result, there is no post-processing cost, and the graphs are extracted in a manner similar to Ringo.
ExtGraph achieves the best performance, in particular outperforms R2GSync by 2.78X at SF 30.

Table~\ref{tab:real_world_graph_extraction_performance} shows the graph extraction performance using real datasets in Figures~\ref{fig:dblp_graph_models} and \ref{fig:imdb_graph_models}. ExtGraph consistently outperforms the other methods, including Ringo. 
However, R2GSync and GraphGen demonstrate notably poor performance in comparison, indicating the inefficiency and significant cost associated with post-processing the decomposed graph structure, as discussed earlier.

\begin{table}[htbp]
\caption{Performance comparison using real datasets (For GraphGen and R2GSync, converting time is specified in parenthesis).}
\label{tab:real_world_graph_extraction_performance}
\vspace{-0.2cm}
\centering
\begin{tabularx}{\linewidth}{c*{1}{Y}c*{4}{Y}}
\toprule
 & \multicolumn{4}{c}{Elapsed Time (sec.)} \\   
 \cmidrule(l){2-5}
Dataset & Ringo & GraphGen & R2GSync & ExtGraph \\
\midrule
DBLP & 769.2 & \begin{tabular}[c]{@{}c@{}}1866.6 \\(1190.1)\end{tabular} & \begin{tabular}[c]{@{}c@{}}1882.1 \\(1214.7)\end{tabular} & \underline{\textbf{669.9}}\\
IMDB & 88.4 & \begin{tabular}[c]{@{}c@{}}116.4 \\(27.3)\end{tabular} & \begin{tabular}[c]{@{}c@{}}388.8 \\(66.2)\end{tabular} & \underline{\textbf{71.2}}\\
\bottomrule
\end{tabularx}
\vspace{-0.3cm}
\end{table}

\subsection{Performance Breakdown}
\label{sec:extgraph_characteristic}

We conducted a performance breakdown experiment to assess the impact of JS-OJ and JS-MV techniques and to evaluate the effectiveness of the hybrid plan combining both techniques. 
For a test graph model, we integrate two scenarios --- recommendation scenario for $Catalog$ and fraud detection scenario for $Store$ --- using the TPC-DS dataset at SF=100.
The resulting edge definitions in Figure~\ref{fig:performance_breakdown}(a) include four queries: $Sell$, $Buy$, $Co\text -pur$, and $Same\text -pro$.
We explored four configurations: no join sharing (both X), JS-OJ, JS-MV, and the hybrid plan (both O).
Figure~\ref{fig:performance_breakdown}(b) shows the hybrid plan generated by Algorithm~\ref{alg:JS_optimization}.

The result of the performance breakdown experiment are presented in Figure~\ref{fig:performance_breakdown}(c).
Here, applying JS-MV alone outperforms applying JS-OJ alone because it optimizes the queries for $Sell$ and $Buy$ edges by reducing redundant operation for $SS \Join I$.
However, JS-OJ is less effective for $Co\text -pur$ and $Same\text -pro$ edges, as discussed in Section~\ref{sec:join_sharing}. 
The hybrid plan, combining both JS-OJ and JS-MV, delivers the best performance among all configurations. 
This is because the hybrid plan does not require the materialization of the query results for $Sell$ and $Buy$ edges, which is managed by its JS-OJ sub-plan (indicated by the red edges in Figure~\ref{fig:performance_breakdown}(b)), and so is more efficient than applying JS-MV alone.
These results confirm that the hybrid plan is highly effective for extracting complex graphs, and our cost model adeptly identifies the most efficient plan for given edge definitions.

\begin{figure}[htbp]
    % \vspace*{-0.3cm}
    \centerline{\includegraphics[width=\linewidth]{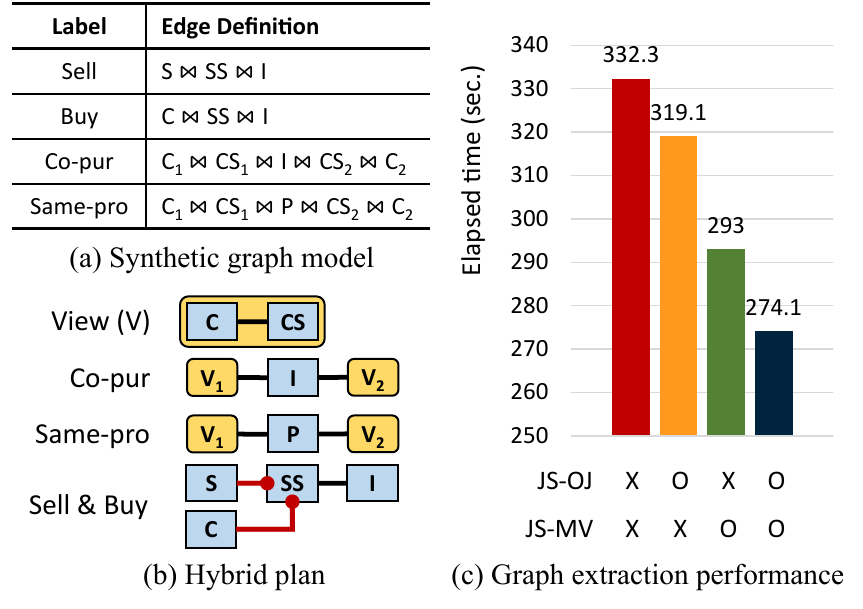}}
    \vspace*{-0.4cm}
    \caption{Results of performance breakdown (SF=30).}
    \label{fig:performance_breakdown}
    \vspace*{-0.2cm}
\end{figure}

\section{Related Works}
\label{sec:related_works}

Schema-based graph extraction leverages referential relationships between tuples in schema information to define graph models. 
Notable methods in this category include R2G~\cite{de2013converting, de2014r2g}, GraphBuilder~\cite{jain2013graphbuilder, willke2012graphbuilder}, Table2Graph~\cite{lee2015table2graph}, GRFusion~\cite{hassan2018extending, hassan2018grfusion}, and Neo4j ETL Tool~\cite{neo4jetl}.
R2G and the Neo4j ETL Tool are designed to automate the definition of graph models from schema. 
GraphBuilder and Table2Graph enhance graph extraction performance by employing the MapReduce model. 
GRFusion, on the other hand, extracts graphs by storing only topological data as an adjacency list, while utilizing relational databases to manage property data. 
However, a common limitation among these methods is that the graph models they can define are significantly constrained by the inherent structure of the referential relationships in the schema.

Join workload-based graph extraction defines graph models by interpreting relationships between tuples through join queries. 
Notable methods include Ringo~\cite{perez2015ringo}, GraphGen~\cite{xirogiannopoulos2015graphgen, xirogiannopoulos2017extracting}, and R2GSync~\cite{anzum2021r2gsync, anzum2019graphwrangler}.
Ringo enables users to create edges from various join queries, offering flexibility in defining diverse graph models. However, executing each join query separately may lead to numerous redundant operations during graph extraction. 
GraphGen aims to reduce both the size of the graph and the time required for its extraction by breaking down complex join queries into virtual edges. 
R2GSync takes a similar approach, simplifying each edge into single joins to hasten graph extraction and facilitate the synchronization of updates from relational databases to graphs. 
Nonetheless, a significant limitation of both GraphGen and R2GSync is that the decomposition of graph structure often results in poorer graph query performance, and edges are restricted to being defined solely from chain queries.

% 분량 제한 넘을 시 이 단락은 삭제
% Apart from graph extraction, various studies have focused on translating and executing graph queries on relational databases. 
% Notable methods of this approach include Vertexica~\cite{jindal2014vertexica}, Grail~\cite{fan2015case}, SQLGraph~\cite{sun2015sqlgraph}, and Cytosm~\cite{steer2017cytosm}.
% Vertexica, Grail, and SQLGraph utilize a specialized schema that stores the vertices and edges of a graph as relational tables. Using this schema, they convert graph queries into relational queries. 
% On the other hand, Cytosm stores a mapping between relational data and graph objects (vertices, edges) as metadata. 
% This allows graph queries to be translated to relational queries across any relational schema, provided that the mapping metadata is defined by the user. 
% However, a common limitation across these methods is that they transform graph traversal operations into join operations, which tends to impair graph query performance.

\vspace{-0.3cm}
\section{Conclusions}
\label{sec:conclusions}

In this paper, we have proposed ExtGraph, a novel method for efficiently extracting user-defined graphs from relational databases. 
ExtGraph consolidates multiple queries that share common joins into a single query through outer joins (JS-OJ) or by creating materialized views for the common joins (JS-MV). 
We have proved the single query combined by JS-OJ produces the same result with the original queries, as in Theorem~\ref{th:JS_OJ_validation}.
ExtGraph employs a hybrid plan, optimized based on our proposed cost models, to further enhance performance. 
The hybrid plan identifies more efficient plans for given edge definitions than applying either JS-OJ or JS-MV alone and becomes more effective as graph model complexity increases.
Our comprehensive experiments with TPC-DS, DBLP, and IMDB datasets have shown that ExtGraph outperforms the state-of-the-art graph extraction methods such as Ringo, GraphGen, and R2GSync. 
Specifically, in the TPC-DS dataset, ExtGraph achieved performance improvements of up to 2.34x in the recommendation scenario and up to 2.78x in the fraud detection scenario.

\bibliographystyle{ACM-Reference-Format}
\bibliography{references}

\end{document}